\documentclass[12pt]{article}
\usepackage{amsfonts}
\usepackage{amsmath}
\usepackage{pb-diagram}
\usepackage[all]{xy}

\usepackage[english]{babel}
\usepackage[T1]{fontenc}

\usepackage[top=30truemm,bottom=30truemm,left=25truemm,right=25truemm]{geometry}

\newtheorem{theorem}{Theorem}
\newtheorem{definition}{Definition}
\newtheorem{lemma}{Lemma}
\newtheorem{remark}{Remark}

\newtheorem{proposition}{Proposition}

\newtheorem{proof}{Proof}

\def\<{\langle}
\def\>{\rangle}

\begin{document}

\centerline{\bf A Construction of Dynamical Entropy on CAR Algebras}

\bigskip\bigskip

\centerline{\sc Kyouhei Ohmura}
\vspace{+1mm}
\centerline{\it Department of Information Sciences,}
\centerline{\it Tokyo University of Science,}
\centerline{\it Noda City, Chiba 278-8510, Japan}
\centerline{E-mail: {\tt 6317701@ed.tus.ac.jp, \quad ohmura.kyouhei@gmail.com}}
\bigskip
\centerline{\sc Noboru Watanabe}
\vspace{+1mm}
\centerline{\it Department of Information Sciences,}
\centerline{\it Tokyo University of Science,}
\centerline{\it Noda City, Chiba 278-8510, Japan}
\centerline{E-mail: {\tt watanabe@is.noda.tus.ac.jp}}

\bigskip\bigskip


\centerline{\bf Abstract }

\bigskip
\noindent
The dynamical entropy on von Neumann algebras defined by Accardi, Ohya and Watanabe (AOW entropy) is a natural noncommutative extension of the classical dynamical entropy. 
On the other hand, quantum spin lattice systems currently used in quantum computing and communication processes  are mathematically described by $C^*$-algebras called CAR algebras. Therefore, in order to obtain the average amount of quantum information and to calculate the uncertainty of the dynamics of quantum spin systems, it is necessary to define dynamical entropy on CAR algebras. In this paper, we formulate dynamical entropy on CAR algebras based on the construction of the AOW entropy. Moreover, we compute the introduced entropy for a $2 \times 2$ matrix algebra case which has relation to the quantum spin system. 


\tableofcontents


\section{Introduction}
The classical dynamical entropy for a measure-preserving invertible transformation of a Lebesgue space has important roles in both pure mathematics and classical information theory. In mathematical side, it shows that two classical dynamical systems are isomorphic. In classical information theory, the dynamical entropy gives the average of infomation amount of an information source \cite{bili}, \cite{kol}.\\
\quad There are several ways to define the dynamical entropy on noncommutative algebras \cite{acnote}, \cite{sto}, \cite{itsuse}. In \cite{aow}, using quantum Markov chains on matrix algebras \cite{ncm}, \cite{qmcuni}, Accardi, Ohya and Watanabe formulated the dynamical entropy on von Neumann algebras. The entropy is called AOW entropy and is a natural noncommutative (or quantum) extension of the classical dynamical entropy. Moreover, due to the simplicity of its formulation, one can easily obtain the average amount of quantum information for quantum dynamical systems using the AOW entropy.\\
\quad Incidentally, in recent years, quantum spin lattice systems are used in quantum computing and quantum communication processes. Since the observables of  spin systems are expressed by the elements of  CAR algebras, it is necessary to define dynamical entropy on this algebra to discuss the dynamics and the average information of spin systems strictly. \\
\quad Therefore, in this paper, we define dynamical entropy on CAR algebras based on the construction method of AOW entropy and the definition of Markov chains on CAR algebras given by Accardi, Fidaleo and Mukhamedov \cite{afm}.\\
\quad We organize the paper as follows. In section 2 we recall the definition of the AOW entropy, namely the dynamical entropy through a quantum Markov chain on matrix algebras. In section 3, we briefly review the definitions of the CAR algebra and the Fermion Fock space which gives the algebra. Section 4 is devoted to the notions of Markov states and chains on CAR algebras. In section 5,  we formulate dynamical entropy on CAR algebras based on the construction of the AOW entropy, the defintion of Markov chains on CAR algebras, and using an Umegaki conditional expectation on CAR algebras. In section 6 we calculate our dynamical entropy for a simple model associated with the two-state system.


\section{AOW entropy}
In this section, we construct dynamical entropy on von Neumann algebras using quantum Markov chains on matrix algebras and a noncommutative extension of measurable partitions of a metric space.\\

\noindent
Let - a Hilbert space: $\mathcal{H}$,\\
- a von Neumann algebra with an identity operator $1_{\mathcal{A}}$ acting on $\mathcal{H}$: $\mathcal{A}$,\\
- the set of all normal states on $\mathcal{A}$: $\mathfrak{S}(\mathcal{A})$,\\
- a *-automorphism on $\mathcal{A}$: $\theta$,\\
- the set of all $d \times d $ matrices on $\mathbb{C}$: $M_d$,\\
- the tensor product of $\mathbb{N}$ copies of $M_d$: $\otimes_{\mathbb{N}}M_d$.\\
Then the triplet $(\mathcal{A}, \mathfrak{S}(\mathcal{A}), \theta)$ describes the dynamics of a quantum system.\\
Furthermore, let $\gamma := \{ \gamma_j \} $ be a finite orthogonal partition of $1_{\mathcal{A}} \in \mathcal{A}$, i.e. 
$$
\sum_{j} \gamma_{j} = 1_{\mathcal{A}}\quad , \quad \gamma_{i}\gamma_{j} = \delta_{ij} \gamma_{j},
$$
and $E_e$ be a completely positive map from $M_d \otimes \mathcal{A}$ to $\mathcal{A}$:
\begin{equation}
E_e (\sum_{i,j} e_{ij} \otimes A_{ij}) = \sum_{i} A_{ii},
\end{equation}
with the matrix unit $e_{ij} \in M_d$. Then a transition expectation \cite{ncm}, \cite{afl} $\mathcal{E}_{\gamma, \theta}: M_{d}\otimes \mathcal{A} \to \mathcal{A}$ with respect to $\theta$ is given by
\begin{equation}
\mathcal{E}_{\gamma , \theta} (a \otimes b) :=  \theta \circ E_e (p_{\gamma , e}^{*} (a \otimes b) p_{\gamma , e}) \quad , \quad a \otimes b \in M_d \otimes \mathcal{A},
\end{equation}
where $p_{\gamma ,e} := \sum_{j} e_{jj}\otimes \gamma_{j}$.
In the above notations, {\it a quantum Markov chain} on $\otimes_{\mathbb{N}} M_d$ is defined by
\begin{equation}
\psi := \{ \varphi , \mathcal{E}_{\gamma , \theta} \} \in \mathfrak{S}(\otimes^{\mathbb{N}} \mathit{M_{d}}) ,
\end{equation}
where $\varphi$ is called the {\it initial distribution} of $\psi$. Then $\psi = \{ \varphi , \mathcal{E}_{\gamma , \theta} \}$ is given by
\begin{equation}
\psi (j_1 (a_1)j_2 (a_2) \cdots j_n (a_n))  = \varphi  (\mathcal{E}_{\gamma , \theta} (a_1 \otimes \mathcal{E}_{\gamma , \theta} ( \mathit{a_{\rm2}} \otimes  \cdots \otimes \mathcal{E}_{\gamma , \theta } ( \mathit{a_n} \otimes 1_{\mathcal{A}}) \cdots ))) \ ,\ n \in \mathbb{N}, a_i \in M_d,
\end{equation}
where $j_k$ is an embedding from $a \in M_d$ into the $k$-th factor of  $\otimes_{\mathbb{N}} M_d$, i.e.
\begin{equation}
j_k (a) := 1 \otimes \cdots \otimes 1 \otimes \overset{k-th}{a} \otimes 1 \otimes \cdots
\end{equation}
Let $\varphi$ be a stationary state. Then there exists unique density operator $\rho$ such that 
$$
\varphi (A) = \mathrm{Tr} \rho A\quad ,\quad \forall A \in \mathcal{A}.
$$
For any $a_1\otimes \cdots \otimes a_n \otimes 1_{\mathcal{A}} \in M_d \otimes \cdots \otimes M_d \otimes \mathcal{A}$, we have
$$
\psi (j_1 (a_1)j_2 (a_2) \cdots j_n (a_n))
= \varphi (\mathcal{E}_{\gamma , \theta} (\mathit{a}_{\mathrm{1}}\otimes \mathcal{E}_{\gamma , \theta} (\mathit{a}_{\mathrm{2}} \otimes \cdots \otimes \mathcal{E}_{\gamma , \theta }(\mathit{a}_n \otimes 1_{\mathcal{A}} \cdots ))))
$$
\begin{eqnarray}
&=& \mathrm{Tr}_{(\otimes_{1}^{n}M_d)  \otimes \mathcal{A}}\it \sum_{i_{\rm1}} \cdots \sum_{i_{n{\rm -1}}} \sum_{i_n} e_{i_{\rm 1} i_{\rm1}}\otimes \cdots \otimes e_{i_{n-{\rm 1}} i_{n-{\rm 1}}} \otimes e_{i_n i_n} \nonumber\\
&\ &\quad \otimes \gamma_{i_n}\theta^{*} (\gamma_{i_{n - 1}} \cdots \theta^{*}(\gamma_{i_{1}}\rho \gamma_{i_1}) \cdots \gamma_{i_{n-1}})\gamma_{i_n} (a_1 \otimes \cdots\otimes a_{n-1} \otimes a_n \otimes  1_{\mathcal{A}})\nonumber \\
&=& \mathrm{Tr}_{(\otimes_{1}^{n}M_d)\otimes \mathcal{A}}\rho_{[0,{\it n}]}(a_{1} \otimes \cdots \otimes a_{n{\rm-1}} \otimes 1_{\mathcal{A}}) \nonumber \\
&=& \mathrm{Tr}_{(\otimes_{1}^{n}M_d)}\it \sum_{i_{\rm 1}} \cdots \sum_{i_{n -{\rm 1}}}\sum_{i_n} (\mathrm{tr}_{\it \mathcal{A}} \it \gamma_{i_n} \theta^{*}(\gamma_{i_{n-{\rm 1}}} \cdots \theta^{*}(\gamma_{i_{\rm 1}} \rho \gamma_{i_{\rm 1}}) \cdots \gamma_{i_{n-{\rm 1}}})\gamma_{i_{n}}) \nonumber \\
&\ &\quad e_{i_1 i_1}\otimes \cdots \otimes e_{i_{n-{\rm 1}} i_{n-{\rm 1}}} \otimes e_{i_n i_n} (a_1 \otimes \cdots \otimes a_{n-1} \otimes a_n) \nonumber\\
&=& \mathrm{Tr}_{(\otimes_{1}^{n}M_d)}\rho_{\it n} (\it a_{\rm 1}\otimes \cdots \otimes a_{n-{\rm 1}} \otimes a_n) ,\nonumber
\end{eqnarray}
where
\begin{eqnarray*}
\rho_{[0,n]} &:=& \sum_{i_1} \cdots \sum_{i_{n-1}} \sum_{i_n} e_{i_1 i_1}\otimes \cdots \otimes e_{i_{n-1} i_{n-1}} \otimes e_{i_n i_n} \\
&\ & \quad \otimes \gamma_{i_n}\theta^{*} (\gamma_{i_n - 1} \cdots \theta^{*}(\gamma_{i_{1}}\rho \gamma_{i_1}) \cdots \gamma_{i_{n-1}})\gamma_{i_n}, \\
\rho_n &:=& \mathrm{Tr}_{\mathcal{A}} \rho_{[0, {\it n}]} .
\end{eqnarray*}
Hence we obtain the density operator $\rho_n$ of $\varphi_n$ on $\otimes^n_1 M_d$. Denoting
\begin{equation}
\Gamma_{i_n \cdots i_1} := \theta^{n-1}(\gamma_{i_n}) \cdots \theta (\gamma_{i_2}) \gamma_{i_1}
\end{equation}
and
\begin{equation}\label{ck}
P_{i_1, \cdots i_n} := \mathrm{Tr}_{\mathcal{A}} \Gamma_{i_n \cdots i_1}\rho \Gamma^{*}_{i_n \cdots i_1} = \mathrm{Tr}_{\mathcal{A}} |\Gamma_{i_n \cdots i_1}|^2 \rho ,
\end{equation}
Accardi, Ohya and Watanabe defined the entropy with respect to $\gamma, \theta$ and $n$ as
\begin{equation}
S_n(\gamma , \theta) := -\mathrm{Tr} \rho_n \log \rho_n = -\sum_{i_{\rm 1} , \cdots , i_n}P_{i_1, \cdots i_n}\log P_{i_1, \cdots i_n }  .
\end{equation}
\begin{definition}
The dynamical entropy through a quantum Markov chain with respect to $\gamma$ and $\theta$ is given by
\begin{eqnarray}
\tilde{S}(\gamma ; \theta) &:=& \limsup_{n \to \infty} \frac{1}{n}S_n (\gamma , \theta) \nonumber \\
&=&  \limsup_{n \to \infty} \frac{1}{n}(-\sum_{i_1 \cdots i_n}P_{i_1, \cdots , i_n}\log P_{i_1, \cdots , i_n}) .
\end{eqnarray}
\end{definition}

\noindent
The above entropy is called {\it AOW entropy}. 

\begin{remark}
Then $P_{i_1 , \cdots , i_n}$ (\ref{ck}) is the time ordered correlation kernel \cite{afl}, \cite{fqp} over $\mathcal{A}$.
\end{remark}

\begin{remark}
The AOW entropy was extended to dynamical mutual entropy and is used to study quantum communication processes \cite{ow}, \cite{muto}.
\end{remark}


\section{CAR Algebras}
The CAR algebra  is a $C^*$-algebra generated by the observables of Fermion systems. In this section, we recall the basic mathematical definitions and physical backgrounds of this algebra \cite{br2}.


\subsection{Fermion Fock spaces}
First, we fix $N \in \mathbb{N}$. Let $\mathcal{H}$ be a Hilbert space of the state of 1-particle. Then the Hilbert space of states of many body systems is give by
\begin{equation}
\mathcal{H}_N := \bigotimes_N \mathcal{H} .
\end{equation}
Then $\mathcal{H}_N$ is the Hilbert space that $N$-particles can be distinguished from each other. Moreover, $\mathcal{P}_N$ denotes a $N$-dimensional group of permutations, i.e.
\begin{equation}
\mathcal{P}_N := \{ \sigma \ ;\ \{ 1, \cdots , N \} \to \{ 1, \cdots , N \} \ ,\ \sigma\ {\rm is\ injective} \}.
\end{equation}
If $\sigma$ is even (resp. odd), we put ${\rm sgn} (\sigma) = 1$ (resp. ${\rm sgn}(\sigma) = -1$) where sgn is sign of $\sigma$.\\

\noindent
For each $\sigma \in \mathcal{P}_N$, there exists a unique unitary operator $P_{\sigma}$ on $\mathcal{H}_N$ satisfying
\begin{equation}
P_{\sigma} (x_1 \otimes \cdots \otimes x_N ) = x_{\sigma (1)} \otimes \cdots \otimes x_{\sigma (N)}\quad , \quad x_j \in \mathcal{H}\ ,\ j \in \{ 1, \cdots , N\}
\end{equation}
and
\begin{equation}
P_{\sigma} P_{\tau} = P_{\tau \sigma}\quad , \quad \sigma \ ,\ \tau \in \mathcal{P}_N .
\end{equation}
$P_{\sigma}$ is called a {\it permutation operator} with respect to a permutation $\sigma$. Using $P_{\sigma}$, we can classify the vectors of $\mathcal{H}_N$.

\begin{definition}
For any $\sigma \in \mathcal{P}_N$,
\begin{enumerate}
\item if $P_{\sigma} x = x$ holds, $x$ is called {\it symmetric}.
\item if $P_{\sigma} x = {\rm sgn}(\sigma) x$ holds, $x$ is called {\it anti-symmetric}.
\end{enumerate}
\end{definition}

\noindent
Put
\begin{equation}
\bigwedge^N (\mathcal{H}) := \{ x \in \mathcal{H}_N\ ;\ x \ {\rm is\ anti-symmetric} \}.
\end{equation}
The elements of $\bigwedge^N (\mathcal{H})$ are called {\it Fermions}.

\begin{remark}
For any $x_1 , \cdots , x_N \in \mathcal{H}$,
$$
x_1 \bigwedge x_2 \bigwedge \cdots \bigwedge x_N := \frac{1}{\sqrt{N!}} \sum_{\sigma \in \mathcal{P}_N} \mathrm{sgn} (\sigma) x_{\sigma (1)} \otimes \cdots \otimes x_{\sigma (N)}
$$
is an anti-symmetric vector. For any $i, j \in \{ 1, \cdots , N \}$ $(i \neq j)$,
$$
x_1 \bigwedge \cdots \bigwedge x_i \bigwedge \cdots \bigwedge x_j \bigwedge \cdots \bigwedge x_N = - x_1 \bigwedge \cdots \bigwedge x_j \bigwedge \cdots \bigwedge x_i \bigwedge \cdots \bigwedge x_N 
$$
holds. Especially, if $x_i = x_j$, one can see that
\begin{equation}\label{pauli}
x_1 \bigwedge x_2 \bigwedge \cdots \bigwedge x_N = 0.
\end{equation}
(\ref{pauli}) implies that two fermions can not exist in the same 1-particle state. The property is called the {\it Pauli exclusion principle}. We will see the algebraic representation of this principle below.
\end{remark}

\noindent
\quad In general, the number of particles in many body systems can change. Therefore one has to consider the  generalized Hilbert space of $\mathcal{H}_N$:
\begin{equation}
\mathcal{F}(\mathcal{H}) := \bigoplus_{N=0}^{\infty} \mathcal{H}_N 
\end{equation}
$$
= \left\{ x = \{ x^{(N)} \}_{N=0}^{\infty} \ ;\ \ x^{(N)} \in \mathcal{H}_N , N \geq 0 ,\ \sum_{N=0}^{\infty} \| x^{(N)}\|_{\mathcal{H}_N}^2 < \infty \right\}
$$
where $\mathcal{H}_0 = \mathbb{C}$. Then the scalar product of $\mathcal{F}(\mathcal{H})$ is defined by
$$
\<x, y \> := \sum_{N=0}^{\infty} \< x^{(N)}, y^{(N)} \>_{\mathcal{H}_N}.
$$
Hilbert space $\mathcal{F}(\mathcal{H})$ is called a {\it full Fock space}.

\begin{definition}
Let $\bigwedge^0 (\mathcal{H}) := \mathbb{C}$.
\begin{equation}
\mathcal{F}_f (\mathcal{H}) := \bigoplus_{N=0}^{\infty} \bigwedge^N (\mathcal{H}) 
\end{equation}
$$
= \left\{ x = \{ x^{(N)}\}_{N=0}^{\infty} \ ;\ x^{(N)} \in \bigwedge^N (\mathcal{H}),\ N \geq 0,\ \sum_{N=0}^{\infty} \| x^{(N)} \|^2 < \infty \right\}
$$
is called a {\it Fermion Fock space}.
\end{definition}
$\mathcal{F}_f (\mathcal{H})$ is a closed subspace of the full Fock space $\mathcal{F}(\mathcal{H})$.
\begin{definition}
A subspace of $\mathcal{F}_f (\mathcal{H})$:
\begin{equation}
\mathcal{F}_{f, 0} (\mathcal{H}) := \{ x \in \mathcal{F}_f (\mathcal{H}) \ ;\ {\rm If \ there\ exists}\ N_0\ {\rm such \ that} \ N \geq N_0,\ x^{(N)} = 0 \} 
\end{equation}
is called a {\it finite particles subspace} and is dense in $\mathcal{F}_f (\mathcal{H})$.
\end{definition}


\subsection{CAR Algebras}
Let $\mathcal{D}$ be a dense subspace of $\mathcal{H}$ and let
\begin{equation} 
\mathcal{F}_{f, fin}(\mathcal{D}) :=  \{ x \in \mathcal{F}_f (\mathcal{H}) \ ;\ x^{(N)} \in A_N (\widehat{\otimes}^N \mathcal{D}),\ N \geq 0, 
\end{equation}
$$
{\rm If \ there\ exists}\ N_0\ {\rm such \ that} \ N \geq N_0,\ x^{(N)} = 0 \} 
$$
where $A_N := \frac{1}{N!} \sum_{\sigma \in \mathcal{P}_N} P_{\sigma}$. Then $\mathcal{F}_{f, fin}(\mathcal{D})$ is called a {\it finite particles subspace on} $\mathcal{D}$ and is dense in $\mathcal{F}_{f}(\mathcal{H})$. \\
For each $x \in \mathcal{H}$, we define the operator $a_+ (f)$ on $\mathcal{F}_f (\mathcal{H})$ as follows:
\begin{equation}
\mathrm{Dom}(a_+ ) := \left\{ x \in \mathcal{F}_f (\mathcal{H})\ ;\ \sum_{N=0}^{\infty} N \| A_N (f \otimes x^{(N-1)}) \|^2 < \infty  \right\},
\end{equation}
\begin{equation}
(a_+ (f) x)^{(N)} := \sqrt{N} A_N (f \otimes x^{(N-1)})\ (N \geq 1)\quad, \quad  (a_+ (f) x)^{(0)} := 0.
\end{equation}
One can know that $b_+ (f)$ is closed. Since $\mathcal{F}_{f, fin} (\mathcal{D}) \subset \mathrm{Dom} (a_+ (f))$ holds, $a_+ (f)$ is densely defined. Therefore the adjoint operator
\begin{equation}
a (f) := a_+ (f)^*
\end{equation}
exists. Furthermore, since $a_+ (f)$ is closed,
\begin{equation}
a (f)^* := a_+ (f).
\end{equation}
$a (f)^*$ is called a {\it Fermion creation oprator} and $a (f)$ is called a {\it Fermion annihilation operator} respectively. \\
Then the following poposition is satisfied. Let $\Omega$ be a vacuum state.

\begin{proposition}
For any $f_j \in \mathcal{H}\ (j=1, \cdots , n)$,
\begin{equation}
(a(f_1)^* a(f_2)^* \cdots a(f_n)^* \Omega)^{(n)} = \sqrt{n!} A_n (f_1 \otimes \cdots \otimes f_n ),
\end{equation}
\begin{equation}
(a(f_1)^* a(f_2)^* \cdots a(f_n)^* \Omega)^{(N)} = 0 \quad ,\quad N \neq 0.
\end{equation}
\end{proposition}
\noindent
The Fermion operators satisfy the following algebraic relations.

\begin{theorem}
For each $f \in \mathcal{H}$, $a(f),\ a(f)^* \in \mathbf{B}(\mathcal{F}_f (\mathcal{H}))$ and the following relation hold:
\begin{equation}
\{ a(f), a(g)^* \} = \< f, g \> \quad , \quad f, g \in \mathcal{H},
\end{equation}
\begin{equation}\label{car2}
\{ a(f), a(g) \} = 0 \quad , \quad \{ a(f)^* , a(g)^* \} = 0 \quad , \quad f, g \in \mathcal{H},
\end{equation}
where $\{ \cdot , \cdot \}$ is the anti commutator $\{ A , B \} := AB + BA$.
\end{theorem}
In (\ref{car2}), if $f=g$, there hold
\begin{equation}
a(f)^2 = 0 \quad , \quad (a(f)^*)^2 = 0.
\end{equation}
These are the algebraic representations of the Pauli exclusion principle (\ref{pauli}).\\

\noindent
Now we state the definition of the CAR algebra.

\begin{definition}
A $C^*$-algebra generated by bounded operators $\{ a(f), a(f)^* \ ;\ f \in \mathcal{H}\}$ is called a {CAR algebra}.
\end{definition}
The CAR algebra is a $C^*$-algebra of observables of Fermion systems.




\section{Markov States and Chains on CAR Algebras}
In section 3, we showed that the CAR algebra describes the Fermi systems. In this section,  we construct Markov chains on CAR algebras following \cite{afm}.\\
Since the Markov chain is a stochastic process, it is natural to consider on $\mathbb{Z}_+ := \mathbb{N} \cup \{ 0 \}$ as follows.\\

\noindent
Let  $\mathcal{A}_{\mathbb{Z}_+}$ be the CAR algebra generated by $\{ a_i , a_i^* \ ;\ i \in \mathbb{Z}_+ \}$. Namely, $\{ a_i , a_i^* \ ;\ i\in \mathbb{Z}_+ \}$ satisfy:
\begin{equation}\label{car1}
(a_i)^* = a_i^* \quad , \quad \{ a_i^* , a_j \} = \delta_{ij} 1,
\end{equation}
\begin{equation}\label{car2}
\{ a_i , a_j \} = \{ a_i^* , a_j^* \} = 0 \quad, \quad i, j \in \mathbb{Z}_+ ,
\end{equation}
where $\{ \cdot , \cdot \}$ is the anti-commutator. For any $I \subset \mathbb{Z}_+$, $\mathcal{A}_I$ denotes the subalgebra generated by $\{ a_i , a_i^* \ ;\ i \in I \}$. In particular,
$$
\mathcal{A}_{n]} := \mathcal{A}_0 \bigvee \mathcal{A}_1 \bigvee \cdots \bigvee \mathcal{A}_{n}.
$$

\begin{remark}
In general, the *-algebra generated by a subset $\mathcal{M} \subset \mathcal{A}$ is given by
$$
\mathrm{l.i.h.}(\{ M_1 \cdots M_n\ ;\ n \in \mathbb{N}, \ M_j \in \mathcal{M} \cup \mathcal{M}^*,\ j= 1, \cdots , n \})
$$
\begin{equation}
= \left\{ \sum_{n=1}^N \lambda_n M_1^{(n)} \cdots M_m^{(n)} \ ;\ N \in \mathbb{N},\ \lambda \in \mathbb{C} \right\}.
\end{equation}
Especially, in CAR case, the *-algebra generated by $\{ a_i , a_i^* \ ;\ i \in \{ k\} \}$ is written as
\begin{equation}
\mathcal{A}_{k} := \left\{ \sum_{n=1}^N \lambda_n A_{i_1}^{(n)} \cdots A_{i_m}^{(n)}\ ;\ A_{i_k}^{(n)} \in \{ a_k ,a_k^* \} ,\ N \in \mathbb{N},\ \lambda_n \in \mathbb{C} \right\} .
\end{equation}
Due to the relations (\ref{car1})-(\ref{car2}), $\mathcal{A}_k$ becomes
\begin{equation}
\mathcal{A}_k = \{ \lambda_1 a_k + \lambda_2 a_k^* + \lambda_3 a_k^* a_k + \lambda_4 a_k a_k^* \ ;\ \lambda_n \in \mathbb{C} \} .
\end{equation}
Therefore, we only consider the elements $\{ a_k , a_k^* , a_k^* a_k , a_k a_k^* \}$ in $I = \{ k\} ,\ k \in \mathbb{Z}_+$ case.
\end{remark}

Put $\mathcal{A} := \mathcal{A}_{\mathbb{Z}_+}$.

\begin{definition}
We denote $\Theta_I$ as the unique automorphism on $\mathcal{A}$ which satisfy
\begin{equation}
\Theta_I (a_i) = -a_i \quad , \quad \Theta_I (a_i^*) = -a_i^* \quad , \quad (i \in I)
\end{equation}
\begin{equation}
\Theta_I (a_i) = a_i \quad , \quad \Theta_I (a_i^*) = a_i^* \quad , \quad (i \in I^c )
\end{equation}
$\Theta_I$ is called a {\it parity automorphism} on $\mathcal{A}$.
\end{definition}
\noindent
Let $\Theta := \Theta_{\mathbb{Z}_+}$. Then $\Theta$ induces even and odd parts of $\mathcal{A}$:
\begin{equation}
\mathcal{A}_+ := \{ a \in \mathcal{A}\ ;\ \Theta (a) = a \} \quad ,\quad \mathcal{A}_- := \{ a \in \mathcal{A}\ ;\ \Theta (a) = -a \} .
\end{equation}

\begin{definition}
Let $\mathcal{A},\ \mathcal{B}$ be CAR algebras and $E$ be a map from $\mathcal{A}$ to $\mathcal{B}$. If
\begin{equation}
E \circ \Theta = E
\end{equation}
holds, $E$ is called {\it even}.
\end{definition}

\noindent
If $E$ is even, 
\begin{equation}
E (a) = E (\Theta (a)) = -E(a) = 0
\end{equation}
holds for each $a \in \mathcal{A}_-$.


\begin{definition}
Let $\varphi$ be a state on $\mathcal{A}$. If there exists a quasi-conditional expectation $E_n$ w.r.t. $\mathcal{A}_{n-1]} \subset \mathcal{A}_{n]} \subset \mathcal{A}_{n+1]}$ which satisfy
\begin{equation}
\varphi_{n]} \circ E_n = \varphi_{n+1]},
\end{equation}
\begin{equation}
E_n (\mathcal{A}_{[n, n+1]}) \subset \mathcal{A}_n ,
\end{equation}
$\varphi$ is called a {\it Markov state}.
\end{definition}

\noindent
Let $\{ \mathcal{E}_n \}_{n \in \mathbb{N}}$ be a sequence of completely positive identity-preserving maps satisfying for each $n \in \mathbb{N}$:
\begin{equation}\label{en}
\mathcal{E}_n : \mathcal{A}_{[0, n+1]} \to \mathcal{A}_{[0, n]}
\end{equation}
\begin{equation}\label{due1}
\mathcal{E}_{n+1}|_{[0, n]} = 1_{[0, n]},
\end{equation}
\begin{equation}\label{comm}
\mathcal{E}_{n+1} \circ \alpha|_{[0, n+1]} = \alpha|_{[0,n]} \circ \mathcal{E}_n,
\end{equation}
where $\alpha$ is the one-step right shift. One can see that (\ref{comm}) gives the following diagram.
\begin{equation}
\xymatrix@C=36pt@R=36pt{
    \mathcal{A}_{[ 0, n+1]} \ar[r]^{\alpha|_{[0, n+1]}} \ar[d]_{\mathcal{E}_n} & \mathcal{A}_{[0, n+2]} \ar[d]^{\mathcal{E}_{n+1}} \\
     \mathcal{A}_{[ 0 , n ]} \ar[r]_{\alpha|_{[0,n]}} & \mathcal{A}_{[ 0, n+1 ]}
  }
\end{equation}
Let $\rho$ be a density operator on $\mathcal{A}_0$ which satisfies the stationarity in the sense of 
\begin{equation}\label{due2}
\rho = \rho \circ \mathcal{E}_0 \circ \alpha|_{0}.
\end{equation}
In the above notations one obtain the sequence of states $\{ \varphi_n \}_{n \in \mathbb{N}}$:
\begin{equation}
\varphi_n := \mathrm{Tr} \rho \circ \mathcal{E}_{0} \circ \cdots \circ \mathcal{E}_{n}.
\end{equation}
Due to (\ref{due1}), (\ref{comm}) and (\ref{due2}), one can extend $\varphi_n$ to a shift-invariant state $\varphi$ on $\mathcal{A}_{\mathbb{N}}$. The shift-invariant state is called the {\it quantum Markov chain} generated by $(\rho, \{ \mathcal{E}_n \}_{n \in \mathbb{N}})$.


\section{A Construction of Dynamical Entropy on CAR Algebras}

\noindent
In this section, giving an Umegaki conditional expectation \cite{moriya}, \cite{condexp}, we formulate dynamical entropy on CAR algebras.\\

\noindent
Let $\mathcal{A}_0$ be a CAR algebra generated by $\{ a_i, a_i^* \ ;\ i \in \{ 0\} \}$ and $\varphi_0$ a faithful normal state on $\mathcal{A}_0$. $\mathcal{E}_{\gamma, \theta}$ denotes an Umegaki conditional expectation from $\mathcal{A}_{\{0 , n\}}$ to $\mathcal{A}_{0}$ as
\begin{equation}
\mathcal{E}_{\gamma , \theta} := \frac{1}{2} \left( \mathrm{Tr}_{n} \left( \Theta^n \otimes \sum_{i=1}^2 \theta^{n-1} (\gamma_i)^* (1_{0}) \theta^{n-1} (\gamma_i) \right) \right)
\end{equation}
where $1_0$ is the identity of $\mathcal{A}_0$, $\theta$ is a *-automorphism on $\mathcal{A}_0$, and $\sum_{i=1} \gamma_i  = 1_0$ is an operator partition of $1_0$. Since $\theta (1_0 ) = 1_0$, $\sum \theta (\gamma_i)$ is again operator partition of $1_0$.\\

\noindent
Now we construct Markov chain with 1-time evolution:
$$
\varphi_0 (\mathcal{E}_{\gamma , \theta} (A_1 \otimes 1_0)) = \mathrm{Tr}_0 \rho_0 (\mathcal{E}_{\gamma , \theta} (A_1 \otimes 1_0))
$$
$$
= \mathrm{Tr}_0 \rho_0 \left( \frac{1}{2} \left( \mathrm{Tr}_1 \left(\Theta^1 \otimes \sum_i \theta (\gamma_i)^* \theta (\gamma_i)\right) (A_1 \otimes 1_0 )\right) \right)
$$
$$
= \frac{1}{2} \mathrm{Tr}_0 \rho_0 \left( \mathrm{Tr}_1 \left( A_{1, +} \otimes \sum_i \theta (\gamma_i)^* \theta (\gamma_i ) \right)\right)
=  \mathrm{Tr}_{[0,1]} \left( \frac{1}{2} A_{1, +} \otimes \sum_i \rho_0 \theta (\gamma_i)^* \theta (\gamma_i) \right).
$$
Since 
$$
A_{k, +} = a_k a_k^* + a_k^* a_k = 1 \quad , \quad \sum_i \theta (\gamma_i)^* \theta (\gamma_i) = \theta (\sum_i \gamma_i) = 1,
$$
we have
$$
\mathrm{Tr}_k \frac{1}{2} A_{k, +} = 1 \quad , \quad \mathrm{Tr}_0 \rho_0 \theta (\gamma_i)^* \theta (\gamma_i) = 1
$$
respectively. Therefore, we obtain a state through a Markov chain on $\mathcal{A}_{[0,1]}$ as
\begin{equation}
\rho_{[0, 1]} := \frac{1}{2} A_{1, +} \otimes \sum_i \rho_0 \theta (\gamma_i)^* \theta (\gamma_i).
\end{equation}

\noindent
Moreover, we consider the case of $n$-time evolution:
$$
\varphi_0 (\mathcal{E}_{\gamma , \theta} (A_1 \otimes \mathcal{E}_{\gamma , \theta} (A_2 \otimes \cdots \mathcal{E}_{\gamma , \theta} (A_{n-1} \otimes \mathcal{E}_{\gamma , \theta} (A_n \otimes 1_0 )) \cdots )))
$$
$$
= \varphi_0 \left( \mathcal{E}_{\gamma , \theta} \left(A_1 \otimes \mathcal{E}_{\gamma , \theta} \left(A_2 \otimes \cdots \mathrm{Tr}_n \left( \frac{1}{2} \left( \left( \Theta^n \otimes \sum_{i_n} \theta (\gamma_{i_n})^* \theta (\gamma_{i_n}) \right) (A_n \otimes 1_0 ) \right)\right) \cdots \right) \right) \right)
$$
$$
= \varphi_0 \left( \mathcal{E}_{\gamma , \theta} \left(A_1 \otimes \mathcal{E}_{\gamma , \theta} \left(A_2 \otimes \cdots \mathrm{Tr}_n \left( \frac{1}{2} \left( A_{n, +} \otimes \sum_{i_n} \theta (\gamma_{i_n})^* \theta (\gamma_{i_n}) \right)\right) \cdots \right) \right) \right)
$$
$$
=\frac{1}{2}\mathrm{Tr}_n (A_{n, +})\cdot \varphi_0 \left( \mathcal{E}_{\gamma , \theta} \left( A_1 \otimes \cdots \mathcal{E}_{\gamma , \theta} \left( A_{n-1} \otimes \sum_{i_n} \theta (\gamma_{i_n})^* \theta (\gamma_{i_n})\right) \cdots \right) \right)
$$
$$
= \frac{1}{2}\mathrm{Tr}_n (A_{n, +})\cdot \varphi_0 \left( \mathcal{E}_{\gamma , \theta} \left( A_1 \otimes \cdots  \mathrm{Tr}_{n-1} \left( \frac{1}{2} \left( A_{n-1, +} \otimes \sum_{i_n , i_{n-1}} \theta (\gamma_{i_{n-1}}^* \theta (\gamma_{i_n}^* \gamma_{i_n}) \gamma_{i_{n-1}}) \right) \right) \cdots \right) \right)
$$
$$
= \frac{1}{2^2}\mathrm{Tr}_{n, n-1}(A_{n , +} \otimes A_{n-1, +}) \cdot \varphi_0 \left( \mathcal{E}_{\gamma , \theta} \left( \cdots  \mathcal{E}_{\gamma , \theta} \left( A_{n-2} \otimes  \sum_{i_n , i_{n-1}} \theta (\gamma_{i_{n-1}}^* \theta (\gamma_{i_n}^* \gamma_{i_n}) \gamma_{i_{n-1}}) \right) \cdots \right)\right)
$$
$$
\vdots
$$
$$
= \frac{1}{2^n} \mathrm{Tr}_{\{n, \cdots ,1\}} (A_{n, +} \otimes A_{n-1, +} \otimes \cdots \otimes A_{1, +}) \cdot \varphi_0 (\sum_{i_n , \cdots , i_1} \theta (\gamma_{i_1}^* \cdots \theta (\gamma_{i_n}^* \gamma_{i_n}) \cdots \gamma_{i_1}))
$$
$$
= \frac{1}{2^n} \mathrm{Tr}_{\{n, \cdots ,1\}} (A_{n, +} \otimes A_{n-1, +} \otimes \cdots \otimes A_{1, +}) \cdot \mathrm{Tr}_0 \rho_0 (\sum_{i_n , \cdots , i_1} \theta (\gamma_{i_1}^* \cdots \theta (\gamma_{i_n}^* \gamma_{i_n}) \cdots \gamma_{i_1}))
$$
$$
= \frac{1}{2^n} \mathrm{Tr}_{\{0, \cdots , n\}}  \left(A_{n, +} \otimes A_{n-1, +} \otimes \cdots \otimes A_{1, +} \otimes \sum_{i_1 , \cdots , i_n} \rho \theta (\gamma_{i_1})^* \theta^2 (\gamma_{i_2})^* \cdots \theta^n (\gamma_{i_n}^* \gamma_{i_n}) \cdots  \theta^2 (\gamma_{i_2}) \theta (\gamma_{i_1}) \right).
$$
Hence a state thorough a Markov chain on $\mathcal{A}_{[0, n]}$ is given by
\begin{equation}
\rho_{[0, n]} := \frac{1}{2^n} \left(A_{n, +} \otimes A_{n-1, +} \otimes \cdots \otimes A_{1, +} \otimes \sum_{i_1 , \cdots , i_n}  \theta^n (\gamma_{i_n}) \cdots  \theta^2 (\gamma_{i_2}) \theta (\gamma_{i_1}) \rho \theta (\gamma_{i_1})^* \theta^2 (\gamma_{i_2})^* \cdots \theta^n (\gamma_{i_n})^* \right).
\end{equation}
Therfore, we get an $n$-time developped state on $\mathcal{A}_{[1, n]}$ as
$$
\rho_n := \frac{1}{2^n} \sum_{i_1 , \cdots , i_n}  \mathrm{Tr}_0 \theta^n (\gamma_{i_n}) \cdots  \theta^2 (\gamma_{i_2}) \theta (\gamma_{i_1}) \rho \theta (\gamma_{i_1})^* \theta^2 (\gamma_{i_2})^* \cdots \theta^n (\gamma_{i_n})^* 
$$
\begin{equation}
\times A_{n, +} \otimes A_{n-1, +} \otimes \cdots \otimes A_{1, +} 
\end{equation}
Now we focus on the constant of $\rho_n$. Denote
\begin{equation}\label{P}
P_{i_1, \cdots , i_n} := \mathrm{Tr}_{0} \theta^n (\gamma_{i_n}) \cdots  \theta^2 (\gamma_{i_2}) \theta (\gamma_{i_1}) \rho \theta (\gamma_{i_1})^* \theta^2 (\gamma_{i_2})^* \cdots \theta^n (\gamma_{i_n})^* .
\end{equation}
Obviously, then there holds
$$
\sum P_{i_1, \cdots , i_n} = 1.
$$
Hence, in the above notations, we can define dynamical entropy on the CAR algebra as following.

\begin{definition}
For a quadruple $(\mathcal{A}, \varphi_0 , \gamma , \theta)$, the {\it dynamical entropy} on the CAR algebra is given by
\begin{equation}\label{dyn}
h_{\varphi_0}(\theta) := \sup_{\gamma} \left\{ - \limsup_{n \to \infty} \frac{1}{n} \sum_{i_1 , \cdots , i_n} P_{i_1 , \cdots , i_n} \log P_{i_1 , \cdots , i_n} \right\} .
\end{equation}
\end{definition}


\section{Model Computation}
In this section, we calculate the complexity of a time-evolution of a simple model associated with the quantum spin system using the introduced dynamical entropy (\ref{dyn}).\\

\noindent
Let 
\begin{equation}
e_0 :=
\begin{pmatrix}
0 \\
1
\end{pmatrix}
\quad , \quad e_1 :=
\begin{pmatrix}
1\\
0
\end{pmatrix}
\end{equation}
be orthonormal system on $\mathbb{C}^2$ and let 
\begin{equation}
a^*_0 :=
\begin{pmatrix}
0 & 1\\
0 & 0
\end{pmatrix}
\quad , \quad
a_0 :=
\begin{pmatrix}
0 & 0\\
1 & 0
\end{pmatrix}
.
\end{equation}
$e_0$ and $e_1$ describe the models of spin-down and spin-up of a quantum spin system respectively.
Furthermore, by the operations:
\begin{equation}
e_0 \overset{a_0^*}{\mapsto} e_1 \overset{a_0^*}{\mapsto} 0 \quad , \quad e_1 \overset{a_0}{\mapsto} e_0 \overset{a_0}{\mapsto} 0 ,
\end{equation}
one can see that $a_0^*$ is the creation operator and $a_0$ is the annihilation operator of this system. Since
$$
a_0^* a_0 =
\begin{pmatrix}
1 & 0\\
0 & 0
\end{pmatrix}
\quad , \quad
a_0 a_0^* =
\begin{pmatrix}
0 & 0\\
0 & 1
\end{pmatrix}
,
$$
$\mathcal{A}_0 := M(2, \mathbb{C})$ is generated by $\{ a_0^*, a_0 \}$. Now we investigate the complexity of a time development on this system using our dynamical entropy. According to
$$
a_0^* a_0 + a_0 a_0^* = 1 \quad , \quad (a_0^* a_0)^* = a_0^* a_0,
$$
we put:\\
An operatior partition of the identity $1 \in \mathcal{A}_0$:
\begin{equation}
\gamma_1 := a_0^* a_0 \quad, \quad \gamma_2 := a_0 a_0^* ,
\end{equation}
a *-automorphism:
$$
\theta (a) := U a U^* \equiv e^{ia_0^* a_0} a e^{-i a_0^* a_0} = 
$$
\begin{equation}
= 
\begin{pmatrix}
e^{i} & 0\\
0 & 1
\end{pmatrix}
a
\begin{pmatrix}
e^{-i} & 0\\
0 & 1
\end{pmatrix}
= (e^{i}a_0^* a_0 + a_0 a_0^* ) a (e^{-i} a_0^* a_0 + a_0 a_0^* ) \quad , \quad a \in \mathcal{A}_0 .
\end{equation}
Moreover, we denote a density matrix by
\begin{equation}
\rho := 
\begin{pmatrix}
\lambda & 0\\
0 & 1-\lambda
\end{pmatrix}
= \lambda a_0^* a_0 + (1 - \lambda) a_0 a_0^* \quad , \quad 0 \le \lambda \le 1.
\end{equation}
Due to
$$
\theta^n (\gamma_{i_n}) \theta^{n-1} (\gamma_{i_{n-1}}) \cdots \theta (\gamma_{i_1}) =
$$
$$
=  U^{n} \gamma_{i_n} U^{* n} U^{n-1} \gamma_{i_{n-1}} U^{* n-1} \cdots U \gamma_{i_1} U^* 
$$
$$
= U^{n} \gamma_{i_n} U^{*} \gamma_{i_{n-1}} U^{*} \cdots U^* \gamma_{i_1} U^* ,
$$
(\ref{P}) becomes
$$
P_{i_1 , \cdots , i_n} = \mathrm{Tr} U^{n} \gamma_{i_n} U^{*} \gamma_{i_{n-1}} U^{*} \cdots U^* \gamma_{i_1} U^* \rho U \gamma_{i_1} U \cdots U \gamma_{i_{n-1}} U \gamma_{i_n} U^{* n}
$$
$$
= \mathrm{Tr} \gamma_{i_n} U^{*} \gamma_{i_{n-1}} U^{*} \cdots U^* \gamma_{i_1} U^* \rho U \gamma_{i_1} U \cdots U \gamma_{i_{n-1}} U \gamma_{i_n}.
$$
Now we mention the following result.
\begin{lemma}
Under the above notations,
\begin{equation}\label{lemofP}
P_{1, 1, \cdots , 1} = \lambda \ ,\ P_{1, 2, \cdots , 1} = \cdots = P_{2, 2, \cdots , 1} = 0 \ ,\ P_{2, 2, \cdots , 2} = 1 - \lambda .
\end{equation}
\end{lemma}
\begin{proof}
According to the anti-commutation relation (\ref{car2}),
$$
U \gamma_1 = (e^i a_0^* a_0 + a_0 a_0^* ) a_0^* a_0 = e^i a_0^* a_0.
$$
The above equation gives
$$
U \gamma_1 U \gamma_1 \cdots \overset{n-th}{U \gamma_1} = (e^i a_0^* a_0)^n = e^{ni} a_0^* a_0 \quad , \quad
\gamma_1 U^* \gamma_1 \cdots \overset{n-th}{\gamma_1 U^*} = (e^{ni} a_0^* a_0)^* = e^{-ni} a_0^* a_0 .
$$
Therefore, if $(i_1 , i_2 , \cdots , i_n) = (1, 1, \cdots , 1)$, there holds
$$
P_{1, 1, \cdots , 1} =  \mathrm{Tr} \gamma_{1} U^{*} \gamma_{1} U^{*} \cdots U^* \gamma_{1} U^* \rho U \gamma_{1} U \cdots U \gamma_{1} U \gamma_{1} =
$$
\begin{equation}
= \mathrm{Tr} e^{-ni} a_0^* a_0 (\lambda a_0^* a_0 + (1 - \lambda) a_0 a_0^* ) e^{ni} a_0^* a_0 
= \mathrm{Tr} \lambda a_0^* a_0 = \lambda. 
\end{equation}
Moreover, 
$$
\gamma_1 U \gamma_2 = a_0^* a_0 (e^i a_0 ^* a_0 + a_0 a_0^* ) a_0 a_0^* = a_0^* a_0 \cdot a_0 a_0^* = 0.
$$
This implies that
\begin{equation}
P_{1, 2, \cdots , 1} = \cdots = P_{2, 2, \cdots , 1} = 0.
\end{equation}
Finally, the equation:
$$
U \gamma_2 = (e^i a_0^* a_0 + a_0 a_0^* ) a_0 a_0^* = a_0 a_0^* 
$$
induces
$$
U \gamma_2 U \gamma_2 \cdots \overset{n-th}{U \gamma_2} = (a_0 a_0^* )^n = a_0 a_0^* \quad , \quad
\gamma_2 U^* \gamma_2 \cdots \overset{n-th}{\gamma_2 U^*} = ((a_0 a_0^*)^* )^n = a_0 a_0^* .
$$
Hence we obtain
\begin{equation}
P_{2, 2, \cdots , 2} = \mathrm{Tr} a_0 a_0^* (\lambda a_0^* a_0 + (1 - \lambda) a_0 a_0^* ) a_0 a_0^* = \mathrm{Tr} (1 - \lambda) a_0 a_0^* = 1 - \lambda . \  \square
\end{equation}

\end{proof}

\noindent
From the above lemma, we then get the dynamical entropy (\ref{dyn}):
\begin{equation}
h_{\rho} (\theta) = -\lim_{n \to \infty} \frac{1}{n} ((1-\lambda) \log (1-\lambda) + \lambda \log \lambda) = 0.
\end{equation}
\begin{remark}
In general, the complexities of unitary operators on a unital $C^*$-algebra are 0 \cite{sto}. Therefore, this result also tells us that our dynamical entropy is defined correctly. 
\end{remark}


\section{Conclusion}
In this paper, based on the construction method of the AOW entropy, we have defined a new type of dynamical entropy on CAR algebras using an Umegaki conditional expectation from $\mathcal{A}_{\{0, n\}}$ to $\mathcal{A}_0$ with a *-automorphisms on a local algebra $\mathcal{A}_0$. Moreover, we have computed the introduced entropy for a $2 \times 2$ matrix algebra case.\\
\quad Incidentally, if one investigate the complexities of the  lattice translation or the Bogoliubov automorphism using dynamical entropy through a Markov chain, we think that it is necessary to reformulate our dynamical entropy with completely positive identity preserving map from $\mathcal{A}_{[0, n+1]}$ to $\mathcal{A}_{[0, n]}$ associated with a *-automorphism on $\mathcal{A}_{\mathbb{Z_+}}$. 

\bigskip

\noindent
{\bf Acknowledgements} The authors would like to thank Prof. Accardi for his useful suggestions and warm encouragements.

\end{document}